\documentclass[12pt,a4paper]{article}

\usepackage{times}

\usepackage{amsthm,amsfonts,amssymb,amsmath}

\usepackage[T1]{fontenc}
\usepackage[cp1250]{inputenc}
\newcommand{\be}{\begin{equation}}
\newcommand{\ee}{\end{equation}}

\newcommand{\ba}{\begin{eqnarray}}
\newcommand{\ea}{\end{eqnarray}}

\usepackage{cite}
\usepackage{color}
\newtheorem{thm}{Theorem}[section]

\newtheorem{prop}[thm]{Proposition}

\newcommand{\p}{\partial}

\newtheorem{rem}{Remark}

\begin{document}
\centerline{\bf  On the energy of a null cone}

\bigskip

-\centerline{J. Tafel}

\noindent
\centerline{Institute of Theoretical Physics, University of Warsaw,}
\centerline{Ho\.za 69, 00-681 Warsaw, Poland, email: tafel@fuw.edu.pl}

\bigskip

\begin{abstract}
 We derive a formula for the Bondi mass aspect in terms of asymptotic data of the Bondi-Sachs metric in the affine gauge. We prove the positivity of the total energy of a regular null cone in agreement with a recent result of Chru\'{s}ciel and Paetz.
\end{abstract}

\bigskip

\noindent
 Keywords:
the Bondi mass aspect, the Trautman-Bondi mass,  null cone

\null

\noindent
PACS numbers: 04.20.Ha

\null

\section{Introduction}
 Schoen and Yau \cite{sy,sy2} proved that the ADM energy in general relativity  is positive under reasonable assumptions about the energy-momentum tensor and metric tensor. The proof was simplified by   Witten \cite{w} and his work was amended  by Nester \cite{n}. Using Witten's method a similar result was obtained  for the Trautman-Bondi energy $E$ of a null surface \cite{lv,sy1,hp,in}.
Recently Chru\'{s}ciel and Paetz \cite{cp} presented  a  proof of the positivity of the Trautman-Bondi energy  of a regular null cone in an asymptotically flat spacetime. This property was  suggested earlier by the present author and Korbicz in a framework of the Hamiltonian formalism on a cone \cite{kt}. However, paper \cite{kt} contains a condition which is too restrictive if the cone is to be regular at its vertex. Here  an amended version  of our approach in \cite{kt} is presented.  In many respects it coincides with that in \cite{cp}, but some proofs and formulas are simpler.

We  consider metric in the Bondi coordinates adapted to a foliation of spacetime by null cones, but instead of the luminosity distance we  use the affine parameter  along null generators of these cones. In section 2 we derive a formula for the Bondi mass aspect in terms of asymptotic data of the metric. In section 3 we express the total energy of a regular cone as an integral of quantities which are either non-negative or become such if the energy dominant condition is satisfied.
 
\section{The Bondi mass aspect in the affine gauge}
In the Bondi-Sachs approach to gravitational radiation \cite{bbm,s} spacetime is asymptotically flat in future null directions and it can be foliated  by null surfaces $u=const$ having the structure of the  Minkowskian light  cone, at least in the asymptotic region. In adapted coordinates $x^{\mu}=u,r,x^A$, where $\mu=0,1,2,3$ and $A=2,3$, metric  takes the form 
\begin{equation}
g=du(g_{00}du+2g_{01}dr+2\omega_{A} dx^A)+g_{AB}dx^Adx^B\ .\label{3}
\end{equation}
Coordinates $u$ and $r$ are interpreted, respectively,  as a retarded time and a distance from a center  and $x^A$ are coordinates on the 2-dimensional sphere $S_2$. We assume that $g$ has the  signature $+---$, so metric $g_{AB}$ is negative definite.
Note that  $g$ corresponds to $\tilde g$ in \cite{kt}.

Within this approach  the following expansions  are assumed for large values of  coordinate $r$ 
\begin{equation}
g_{00}=1-2Mr^{-1}+O(r^{-2})\ ,\label{4}
\end{equation}
\begin{equation}
g_{01}=1+O(r^{-2})\ ,\label{4a}
\end{equation}
\begin{equation}
\omega_{A}=\psi_A+\kappa _A r^{-1}+O(r^{-2})\ ,\label{6}
\end{equation}
\begin{equation}
 g_{AB}=-s_{AB}r^2+ n_{AB}r+m_{AB}+O(r^{-1})\ .\label{7}
\end{equation}
Here $s_{AB}$ is the standard metric of $S_2$ and all coefficients are functions of $u$ and $x^A$. A lack of a term proportional to  $r^{-1}$ in (\ref{4a}) follows from minimal assumptions about the Ricci tensor which assure that the total energy-momentum vector is well defined \cite{t}.

Coordinate $r$ is not yet fully defined. In the original Bondi coordinates it is the luminosity distance  $r'$   which satisfies 
\be\label{1}
r'=(\frac{\sigma}{\sigma_s})^{\frac 12}\ ,
\ee
where $\sigma^2=\det{g_{AB}}$ and $\sigma^2_s=\det{s_{AB}}$. Let quantities related to the Bondi coordinates be denoted by prime. Function  $M'$ is called the Bondi mass aspect and the total energy of the cone is given by 
\be
E(u)=\frac{1}{4\pi}\int_{S_2} M' \sigma_sd^2x\ .\label{10a}
\ee
In these coordinates $s^{AB}n'_{AB}=0$ and $n'_{AB,u}$ is equivalent to the Bondi news function (this is why we use letter  $n$ in the symbol $n_{AB}$).

Unfortunately, the luminosity distance is not convenient for proving $E\geq 0$. It is better to assume that  
\be
 g_{01}=1\ .\label{12}
\ee
 In this case  $r$ is the affine parameter along null geodesics  tangent to  $u^{,\alpha}\partial_{\alpha}$. A transformation between the luminosity distance $r'$ and the affine distance $r$ is defined  up to a function $a'(u,x^A)$
\be\label{12a}
r=r'+\int_{r'}^{\infty}(1-g_{01})dr''+a'(u,x^A)\ .
\ee
If $a'\neq 0$ then instead of (\ref{4}) one obtains 
\begin{equation}
g_{00}=a-2Mr^{-1}+O(r^{-2})\ ,\label{13}
\end{equation}
where $a$ is a  function of $u$ and $x^A$. A relation between $M$ and the Bondi mass aspect $M'$ is described by the following proposition in which $\p_u$ is denoted by a dot.
\begin{thm} 
Metric
\begin{equation}
g=du(g_{00}du+2dr+2\omega_{A} dx^A)+g_{AB}dx^Adx^B\label{13a}
\end{equation}
is equivalent to the Bondi-Sachs metric with the luminosity distance if it satisfies  conditions   (\ref{6}), (\ref{7}) and 
\begin{equation}
g_{00}=1+\frac 12\dot n^A_{\ A}-\frac{2M}{r}+O(r^{-2})\ .\label{18a}
\end{equation}
The Bondi mass aspect is given by
\be
M'=M+\frac{1}{4}\dot m^A_{\ A}-\frac 14n_{AB} \dot n^{AB}+\frac {1}{16}n^A_{\ A}\dot n^{\ B}_B\ .\label{18g}
\ee
 If the Einstein equation $R_{11}=T_{11}$ is satisfied then
\be
 m^A_{\ A}=\frac 14n_{AB} n^{AB}-\lim_{r\rightarrow\infty}{r^4 T_{11}}\ .\label{18f}
\ee
\end{thm}
\begin{proof}
Given metric (\ref{13a}) function $g'_{00}$ and  the Bondi mass aspect $M'$ can be determined  from the equality
\begin{equation}
g'_{00}du+2g'_{01}dr'+2\omega'_{A} dx^A=g_{00}du+2dr+2\omega_{A} dx^A\ ,\label{14}
\end{equation}
where primes correspond to the luminosity gauge.
It follows from (\ref{14}) that
\be
g_{01}'=\frac{1}{r'_{,r}}\ ,\ \ \ g_{00}'=g_{00}-2\frac{r'_{,u}}{r'_{,r}}\ .\label{16}
\ee
 In virtue of  (\ref{1})  equation (\ref{16}) yields
\be
g_{00}'=g_{00}-2\frac{\sigma_{,u}}{\sigma_{,r}}\ .\label{17}
\ee
From (\ref{7}) and (\ref{1}) one obtains
\be
\sigma\approx \sigma_sr^2\big(1+\frac {1}{2r}n^A_{\ A})\ ,\ \ r'\approx r+\frac {1}{4}n^A_{\ A}\label{18}
\ee
for large values of $r$. One consequence of   (\ref{17}), (\ref{18})  and condition (\ref{4}) for $g'_{00}$ is equation (\ref{18a}).  Another one  is that  the Bondi mass aspect is given by
\be
M'=\frac 12\lim_{r\rightarrow\infty}{r'(1-g_{00}+2\frac{\sigma_{,u}}{\sigma_{,r}})}\label{18i}
\ee
or, equivalently, by
\be
M'=\frac {1}{4\sigma_s}\lim_{r\rightarrow\infty}{(\sigma_{,r}(1-g_{00})+2\sigma_{,u})}=\frac {1}{4\sigma_s}\lim_{r\rightarrow\infty}{f}+\frac 18n^A_{\ A}\ ,\label{18b}
\ee
where
\be
f=2\sigma_sr-g_{00}\sigma_{,r}+2\sigma_{,u}\ .\label{18j}
\ee

Expression (\ref{18b}) will be useful in a proof of $E\geq 0$ in the next section. In order to obtain (\ref{18g}) we calculate one more term in  expansions   (\ref{18})
\be
\sigma\approx \sigma_s(r^2+\frac 12rn^A_{\ A}+\frac 12\tilde m^A_{\ A}+\frac {1}{16}(n^A_{\ A})^2) ,\ \ r'\approx r+\frac {1}{4}n^A_{\ A}+\frac{1}{4r}\tilde m^A_{\ A}\ ,\label{18c}
\ee
where
\be
\tilde m^A_{\ A}=m^A_{\ A}-\frac 12n_{AB}n^{AB}+\frac 18(n^A_{\ A})^2\ .\label{18d}
\ee
It follows from (\ref{18b})-(\ref{18d}) that 
\be
M'=M+\frac{1}{4}\dot{\tilde m}^A_{\ A}\ ,\label{18e}
\ee
hence (\ref{18g}) is obtained. 
One can further transform this expression if the Einstein equation $R_{11}=T_{11}$ (see (\ref{19}))  is satisfied. Expanding both sides of this equation into powers of $r^{-1}$ yields  (\ref{18f}) as the first nontrivial equality. 
 
\end{proof}
\begin{rem}
 In order to eliminate $m^A_{\ A}$ from (\ref{18g}) one can use (\ref{18f}) or the equality
\be
 \dot m^A_{\ A}=\frac 12n_{AB} \dot n^{AB}+\lim_{r\rightarrow\infty}{r^3(2T_{01}}-T)
\ee
which follows from the Einstein equation $R_{01}=T_{01}-T/2$. 
\end{rem}
Formulas given in this section agree with those in \cite{kt} if $T^{\mu\nu}=0$ and $n^A_{\ A}=0$. The latter condition can be obtained by a shift of the affine distance. However,   assumption $n^A_{\ A}=0$ is too restrictive (as pointed out by P. Chru\'sciel) if   a regular foliation of the complete cone $u=const$ by surfaces $r=const$ is considered in order to prove $E\geq 0$. 
\section{The total energy of a light cone}
Let us assume that metric satisfies conditions (\ref{3})-(\ref{7}) and $u=const$ is a complete  future cone with a vertex at $r=0$. For our purposes the following regularity conditions at the vertex are important
\be\label{2}
\lim_{r\rightarrow 0}{r^{-2}\sigma}\neq 0\ ,\ \ \sigma_{,u}=\sigma_{,r}=0\ ,\ \ |g_{00}-\omega_A\omega^A|<\infty\ \ at\ \ r=0\ .
\ee
These conditions are  preserved if we pass to the luminosity distance $r'$ or to the affine distance  provided that the latter  vanishes at the vertex.

We are going to prove the following proposition which, to big extent, overlaps with equation (31) in \cite{cp} and its consequence saying that $E\geq 0$ ($m_{TB}\geq 0$ in the notation of \cite{cp}) if the dominant energy condition is satisfied. The proof of this proposition is  based on equations and method presented by the author and  Korbicz in \cite{kt}. 

\begin{prop}
 Let metric (\ref{13a}) satisfies conditions (\ref{6}), (\ref{7}), (\ref{18a}), (\ref{2})  and the Einstein equations $R_{11}=T_{11}$ and $g^{AB}R_{AB}=g^{AB}T_{AB}-T$. Then
\ba\label{25c}
16\pi E=-\int_{0}^{\infty}dr\int_{S_2}{(\frac{1}{4}r^2\sigma_s\hat g_{AB,r}\hat g^{AB}_{\ \ ,r}+\frac{1}{2}\sigma g_{AB}\omega^A_{\ ,r}\omega^B_{\ ,r}) d^2x}\\\nonumber
+\int_{0}^{\infty}dr\int_{S_2}(r^2\sigma_sT_{11}+\sigma (T-g^{AB}T_{AB}))d^2x\ ,
\ea
where $\hat g_{AB}=r^{-2} g_{AB}$. If the dominant energy condition is satisfied then $E\geq 0$.
\end{prop}
\begin{proof}
The Einstein equations mentioned above have the following form
\be
\frac{1}{4}g_{AB,r}g^{AB}_{\ \ ,r}-(\ln \sigma)_{,rr}=T_{11}\ ,\label{19}
\ee
\be
R^{(2)}+\sigma^{-1}(2\sigma\omega^A_{\ |A}-g^{rr}\sigma,_r-2\sigma,_u),_r-(\omega^A_{\ ,r})_{|A}-\frac{1}{2}g_{AB}\omega^A_{\ ,r}\omega^B_{\ ,r}=g^{AB}T_{AB}-T.\label{20}
\ee
Here ${}_{|A}$ is the covariant derivative with respect to   $g_{AB}$, 
$R^{(2)}$ is the Ricci scalar  of this metric, $T^{\mu\nu}$ is  the energy-momentum tensor and $T=T^{\mu}_{\ \mu}$. Equations (\ref{19}) and (\ref{20}) reduce to equations (29) and (32)  in \cite{kt} if $T^{\mu\nu}=0$.

Let us integrate equation (\ref{20}) over a part of the cone $u=const$ between $r_1$ and $r_2$ with the measure $\sigma dx^2dx^3dr$. Integrals of divergences of $\omega^A$ and $\omega^A_{\ ,r}$ vanish. The Gauss-Bonet theorem for the negative definite metric $g_{AB}$ yields 
\be
\int_{S_2}{R^{(2)}\sigma d^2x}=-8\pi\label{21}
\ee
and it allows to  replace $R^{(2)}\sigma$  by $-2(\sigma_sr)_{,r}$. Hence, 
\be\label{22}
\int_{S_2}(\tilde f(r_2)-\tilde f(r_1))d^2x=\int_{r_1}^{r_2}dr\int_{S_2}P\sigma d^2x\ ,
\ee
where
\be\label{24}
\tilde f=2\sigma_sr+g^{11}\sigma,_r+2\sigma,_u
\ee
and
\be\label{23}
P=T-g^{AB}T_{AB}-\frac{1}{2}g_{AB}\omega^A_{\ ,r}\omega^B_{\ ,r}
\ee
(definition (\ref{24}) coincides with (50) in \cite{kt}). Function  $\tilde f$ vanishes at $r=0$ in virtue of conditions (\ref{2}). Since $g^{11}=\omega^A\omega_A-g_{00}$, $\tilde f$ can be approximated by function (\ref{18j}) for large values of $r$. In virtue of (\ref{18b})  equation (\ref{22}) in the limit $r_1\rightarrow 0$ and $r_2\rightarrow\infty$ yields  the following expression for the total energy 
\be\label{25}
16\pi E=\frac 12\int_{S_2}{n^A_{\ A}\sigma_sd^2x}+\int_{0}^{\infty}dr\int_{S_2}P\sigma d^2x\ .
\ee
Equation (\ref{25})  agrees with (48) in \cite{kt} if $n^A_{\ A}=T_{\mu\nu}=0$.

 In order to find sufficient  conditions  which assure that the first term on the rhs of (\ref{25}) is non-negative let us write equation (\ref{19}) in the following way 
\be
\frac{1}{4}r^2\hat g_{AB,r}\hat g^{AB}_{\ \ ,r}-(r^2(\ln \hat\sigma)_{,r})_{,r}=r^2T_{11}\ ,\label{26}
\ee
where $g_{AB}=r^2\hat g_{AB}$ and $\sigma=r^2\hat \sigma$. From (\ref{18}) and (\ref{2}) one obtains
\be\label{27}
\lim_{r\rightarrow\infty}r^2(\ln \hat\sigma)_{,r}=-\frac 12n^A_{\ A}\ ,\ \ \ \lim_{r\rightarrow 0}r^2(\ln \hat\sigma)_{,r}=0\ .
\ee
Hence, integrating (\ref{26}) over $r$ between $0$ and $\infty$ yields
\be
n^A_{\ A}=2\int_0^{\infty}{r^2(T_{11}-\frac{1}{4}\hat g_{AB,r}\hat g^{AB}_{\ \ ,r}})dr\ .\label{28}
\ee
Substituting (\ref{23}) and (\ref{28}) into (\ref{25}) yields (\ref{25c}).
Since
\be
\hat g_{AB,r}\hat g^{AB}_{\ \ ,r}=-\hat g_{AB,r}\hat g^{AC}\hat g^{BD}\hat g_{CD,r}\leq 0
\ee
condition  $T_{11}\geq 0$  is sufficient to obtain $n^A_{\ A}\geq 0$. 

 Let us consider the integral of $P$ on the rhs of (\ref{25}).  Since $g_{AB}\omega^A_{\ ,r}\omega^B_{\ ,r}\leq 0$, in order  to obtain $P\geq 0$ it is sufficient to assure positivity of  the expression
\be
T-g^{AB}T_{AB}=2T_{10}+g^{11}T_{11}-2\omega^AT_{1A}=T^0_{\ \beta}v^{\beta}=g(\p_r,T^{\alpha}_{\ \beta}v^{\beta}\p_{\alpha})\ ,
\ee
where
\be
 v=2\p_u+(\omega^B\omega_B-g_{00})\p_r-2\omega^A\p_A\ ,\ \ v_{\alpha}v^{\alpha}=-4\omega_A\omega^A\geq 0\ .
\ee
Let $\omega_A\omega^A\neq 0$. Then $v$ is timelike and future directed since $\p_r$ is null future directed and $g(v,\p_r)>0$. If $T^{\alpha\beta}$ satisfies the dominant energy condition then vector $v'=T^{\alpha}_{\ \beta}v^{\beta}\p_{\alpha}$ is non-spacelike and $g(v',v)\geq 0$.  Hence, $v'$ is also future directed and $T^0_{\ \beta}v^{\beta}=g(v',\p_r)\geq 0$. If $\omega_A\omega^A= 0$ then the same result follows from the continuity. Thus, $T-g^{AB}T_{AB}\geq 0$ and $P\geq 0$. 

If the energy dominant condition is satisfied then also $T_{11}\geq 0$  since $T_{11}=T_{\alpha\beta}k^{\alpha}k^{\beta}$, where $k=\p_r$.  Hence, $n^A_{\ A}\geq 0$ and the total energy $E$ is non-negative. 

\end{proof}

\section{Concluding remarks}

Results of Korbicz and Tafel \cite{kt} on  energy of a null surface in the Bondi-Sachs formalism  have been completed.
We presented an expression for the Bondi mass aspect in terms of asymptotical data corresponding to   metric in  the affine gauge (proposition 2.1). Using only two of the Einstein equations we confirmed a result of Chru\'{s}ciel and Paetz \cite{cp} on the positivity of the Trautman-Bondi energy of a regular null cone (proposition 3.1). Our proof is shorter and expression (\ref{25c}) for the total energy  is  simpler than that in \cite{cp}.

\bigskip
\noindent {\bf Acknowledgements}. I am grateful to Piotr Chru\'sciel for fruitful discussions.

\end{document}